\documentclass[letter paper, 10 pt, conference]{ieeeconf}

\IEEEoverridecommandlockouts                              

\usepackage[utf8]{inputenc}
\usepackage{amsmath, amssymb}
\usepackage{graphicx}
\usepackage{color}
\usepackage{ulem}
\usepackage{subcaption}

\renewcommand{\d}{\mathrm{d}}
\newcommand{\up}{u_p}
\newcommand{\ue}{u_e}
\newcommand{\mup}{\mu_p}
\newcommand{\mue}{\mu_e}

\newcommand{\Tp}{\mathcal{T}}
\newcommand{\I}{\mathcal{I}}

\newcommand{\T}{^{\mbox{\tiny \sf T}}}

\newcommand{\R}{\mathbb{R}}

\DeclareMathOperator{\E}{\mathbb{E}}

\newtheorem{theorem}{Theorem}
\newtheorem{lemma}{Lemma}
\newtheorem{corollary}{Corollary}
\newtheorem{remark}{Remark}

\newtheorem{assumption}{Assumption}
\newtheorem{example}{Example}

\renewenvironment{proof}{\textit{Proof:}}{\hfill $\blacksquare$}

\title{\LARGE\bf Efficient Communication for Pursuit-Evasion Games with Asymmetric Information}

\author{Dipankar Maity
 \thanks{D. Maity is with the Department of Electrical and Computer Engineering, University of North Carolina at Charlotte,  NC, 28223, USA. Email:
 		{\small \texttt{dmaity@uncc.edu}}}
\thanks{The work was supported, in parts, by by  ARL grant ARL DCIST CRA W911NF-17-2-0181. }
   }

\date{}

\begin{document}

\maketitle

\begin{abstract}
     We consider a class of pursuit-evasion differential games in which the evader has continuous access to the pursuer's location but not vice-versa. There is an immobile sensor (e.g., a ground radar station) that can sense the evader's location and communicate that information intermittently to the pursuer. Transmitting the information from the sensor to the pursuer is costly and only a finite number of transmissions can happen throughout the entire game. The outcome of the game is determined by the control strategies of the players and the communication strategy between the sensor and the pursuer. We obtain the (Nash) equilibrium control strategies for both the players as well as the optimal communication strategy between the static sensor and the pursuer.
     We discuss a dilemma for the evader that emerges in this game.
     We also discuss the emergence of implicit communication where the absence of communication from the sensor can also convey some actionable information to the pursuer.
\end{abstract}

\section{Introduction}

Pursuit-Evasion games \cite{isaacs1999differential}  have been applied to investigate a wide class of civilian and military applications involving multi-agent interactions in adversarial scenarios \cite{yan2016multi,  robotics9020047, inproceedings}. 
While several variations ranging from complex dynamic models for the players (e.g., \cite{sun2017multiple}) to complex geometry of the environment (e.g., \cite{oyler2016pursuit}) to limited visibility of the players (e.g., \cite{bhattacharya2010existence}) have been considered, one of the prevailing assumptions have been the continuous sensing capability for the players, with the exception of \cite{aleem2015self, maity2016strategies, huang2021defending} among few others. 
By `continuous sensing' we refer to the capability that enables the players to keep their sensors turned on continuously for the entire duration of the game.
Extensions of pursuit-evasion games in the context of sensing limitations have mainly considered limited sensing range e.g., \cite{durham2010distributed}, limited field of view e.g.,  \cite{gerkey2006visibility}, and noisy measurements e.g., \cite{bagchi1981linear}, but the challenges associated with the lack of continuous sensing remains unsolved.  

In this paper, we  revisit the classical linear-quadratic pursuit-evasion game where the pursuer does not have a continuous sensing capability. 
In particular, the pursuer relies on a remotely located sensor (e.g., a radar station) to sense the evader's position. 
Upon request, the remote sensor can perfectly sense the location of the evader and share it with the pursuer.\footnote{
One may alternatively also consider a scenario where the pursuer has an onboard sensor to measure the evader's location, however, it can only use the sensor intermittently due to, for example, energy and computational constraints.
} 
The communication channel between the pursuer and the remote sensor is assumed to be noiseless, instantaneous (i.e., no delay), and perfectly reliable (i.e., no packet losses). 
The pursuer intermittently requests the evader's location to update its pursuit strategy. 
Due to the resource (e.g., energy) constraints, the pursuer can only make a finite number of requests and it aims to minimize the number of communications.
On the other hand, the evader is able to sense the pursuer continuously and is aware of the sensing limitation of the pursuer. 
The objective of this work is to analyze the game under this asymmetric sensing limitation and obtain: (i) the optimal communication times between the sensor and the pursuer and, (ii) the equilibrium control strategies for the pursuer and the evader. 
It is noteworthy that the majority of the existing work not only considers continuous sensing (or no sensing at all) but also often assumes that the sensing capability is superior for the pursuer; e.g., perfect measurement for the pursuer and noisy measurement for the evader \cite{behn1968class, rhodes1969differential}.

To the best of our knowledge, some of the earliest works involving intermittent measurements for linear quadratic differential games were discussed in \cite{maity2016optimal, maity2016strategies}  where both the players had access to only intermittent measurements. 
These players, however, had to come to an agreement on the sensing instances, which was solved by an optimization problem.
This restriction was later relaxed in \cite{maity2017linear}.
These works were also extended to discrete time \cite{maity2017linear}, infinite horizon \cite{maity2017asymptotic},  and recently to asset defense scenarios \cite{huang2021defending}. 
In this paper, we consider a three-player problem involving a pursuer, an evader, and a remote sensor, where the remote sensor is a passive player that does not make any decision about the sensing times. 
The pursuer intermittently communicates with the remote sensor which results in an intermittent sensing capability for the pursuer. 
This problem has similarity with the sampled-data $H_\infty$ optimal control problem studied in \cite{bacsar1991game}. 
In this paper, we not only consider an intermittent sensing (equivalent to the sampled-data) framework, but also we obtain the optimal time instances for the sensing/communication. 

Intermittent sensing/communication has been a subject of active research in networked control systems. 
Efficient frameworks such as event- and self-triggered controls have been developed to reduce the communication/sensing needs. 
These methods have also been used for studying games, e.g., \cite{cai2020distributed, liu2023predefined}. 
Although these methods use intermittent sensing/communication,  they do not study the optimality of the sensing/communication strategy, except in some special cases, e.g., \cite{molin2009lqg, maity2019optimal}. 
This is primarily due to the fact that intermittent sensing/communication makes the problem analytically and computationally intractable. 
In contrast to those works, our objective in this paper is to find the optimal number of required communications as well as the optimal communication times.


The rest of the paper is organized as follows: In Section~\ref{sec:formulation}, we describe the linear-quadratic pursuit-evasion game problem.
In Section~\ref{sec:closed-and-open-loop}, we discuss the closed- and open-loop equilibria for this game. 
The pursuer has continuous access to the game state in the closed-loop case whereas, in the open-loop case, the pursuer only knows the initial state of the game. 
Next, we study the game with intermittent communication in Section~\ref{sec:gameAnalysis}. 
We derive the optimality conditions for the communication as well as the equilibrium control strategies for the pursuer and the evader. 
Finally, we discuss some open problems  and conclude the paper in Section~\ref{sec:conclusion}.

\section{Problem Formulation} \label{sec:formulation}

We consider a linear quadratic pursuit-evasion game where the state of the game follows the dynamics
\begin{align} \label{eq:dyn}
    \dot{x}(t) = A x(t) + B \up(t) + C \ue(t),\qquad x(t_0) = x_0,
\end{align}
where $x(t)\in \R^{n_x}$ is the state at time $t$, $\up(t)\in \R^{n_p}$ and $u_e(t) \in \R^{n_e}$ are the inputs of the pursuer and the evader, respectively, at time $t$. 
The initial state $x_0$ is known to both the players at the beginning of the game.
The payoff function for this game is 
\begin{align} \label{eq:Jbar}
\bar{J} &= \int_{t_0}^{t_f} \big( \|x(t)\|^2_{Q} + \|\up(t)\|^2_{R_p} - \|\ue(t)\|^2_{R_e}\big) \d t \nonumber  \\
&\quad + \|x(t_f)\|^2_{Q_f},   
\end{align}
where $R_p, R_e \succ 0$, and  $Q, Q_f \succeq 0$ are symmetric matrices.


Although the game is deterministic, the payoff may not necessarily be deterministic if the players adopt randomized strategies for their inputs. 
Therefore, in the subsequent section, wherever it is appropriate, we will consider the expected value of the payoff function, i.e., $J \triangleq \E_{\sim(\mup, \mue)}[\bar J]$, where $\mu_p$ and $\mu_e$ denote the strategies of the pursuer and the evader, respectively. 
$\E_{\sim(\mup, \mue)}[\cdot]$ denotes expectation with respect to the randomization induced by $\mup$ and $\mue$.

The strategies $\mup$ and $\mue$ are measurable functions of the information sets of the players. 
To describe the information sets of the players, we first denote $m(t)$ to be the total number of sensing requests upto time $t$ and $\Tp(t) \triangleq \{t_0, t_1, \ldots, t_{m(t)}\}$ be the set of sensing instances upto time $t$, where $t_0$ is the initial time,  and $t_i< t_{i+1}$ for all $i$, and $t_{m(t)}$ is the latest sensing time such that $t_{m(t)}\le t$.
Furthermore, we denote $\I_e(t) \triangleq \{x(s), \Tp (t) ~|~ s\le t \}$ to be the information available to the evader at time $t$ and $\I_p(t) \triangleq \{x(s'), \Tp(t)~|~s'\in \Tp(t)\}$ to be the pursuer's available information.

In the subsequent sections we will suppress the time argument, as much as possible, for notational brevity, e.g., we shall use $x$ instead of $x(t)$. 

\section{Closed-loop and Open-loop Equilibria} \label{sec:closed-and-open-loop}

In this section, we revisit the classical results on close-loop  equilibrium and also discuss the open-loop case for the pursuer. 
To that end, let matrix $P$ follow the Riccati equation
\begin{align} \label{eq:riccati}
    &\dot{P} + A\T P + PA + Q + P(CR_e^{-1}C\T - BR_p^{-1} B\T)P = 0, \nonumber \\
    &P(t_f) = Q_f.
\end{align}
The solution to the Riccati equation is  well defined for all $t\le t_f$ if the assumption $CR_e^{-1}C\T - BR_p^{-1} B\T \prec 0$ is satisfied. 
Without this condition, the Riccati equation may have a finite escape time (conjugate point) in the interval $[t_0, t_f]$ and the  cost $\bar J$ becomes infinity. 
\begin{assumption} \label{asm:controllability}
The matrices are chosen such that $CR_e^{-1}C\T - BR_p^{-1} B\T \prec 0$.
\end{assumption}

Loosely speaking, the above assumption implies that the pursuer has more controllability than the evader, as also discussed in \cite{behn1968class}.

Provided the solution to \eqref{eq:riccati} is well defined for the interval $[t_0, t_f]$ --which is now guaranteed by Assumption~\ref{asm:controllability}-- one may verify  that $\bar{J}$ in \eqref{eq:Jbar} can be rewritten as (e.g., see \cite[Theorem~3.1]{maity2016strategies})
\begin{align} \label{eq:Jbar_squared}
    \bar{J} = &\|x_0\|^2_{P_0} + \int_{t_0}^{t_f} \|\up + R_p^{-1}B\T P x\|^2_{R_p}\d t- \nonumber \\
    & \int_{t_0}^{t_f}\|\ue - R_e^{-1}C\T P x\|^2_{R_e} \d t,
\end{align}
where we defined $P_0 \triangleq P(t_0)$. 
When both the players have access to the state for all $t$, we observe from \eqref{eq:Jbar_squared} that the (Nash) equilibrium inputs for the players are 
\begin{align*}
   & \up  = - R_p^{-1}B\T P x \triangleq \mup(\I_p) , \\
    & \ue = R_e^{-1}C\T P x  \triangleq \mup(\I_e).
\end{align*}
Substituting these strategies in the dynamics \eqref{eq:dyn} yields
\begin{align} \label{eq:xSol}
    x(t) = \Phi(t) x_0 
\end{align}
where the state transition matrix $\Phi$ satisfies
\begin{align} \label{eq:Phi}
    \dot\Phi = (A- BR_p^{-1}B\T P  + CR_e^{-1}C\T P ) \Phi,
\end{align}
with the initial condition $\Phi(t_0) = {I}$, where $I$ is the identify matrix.  
Let us define the `open-loop' input pair $(\bar\up, \bar\ue)$:
\begin{align} \label{eq:barStrategy}
\begin{split}
   & \bar\up = - R_p^{-1}B\T P \Phi x_0 , \\
    & \bar\ue = R_e^{-1}C\T P \Phi x_0.
    \end{split}
\end{align}
Using \eqref{eq:xSol} and \eqref{eq:Phi} we notice that $\bar\up(t) =  \up(t)$ and $\bar\ue(t) =  \ue(t)$ for all $t$. 
Therefore, it appears from \eqref{eq:barStrategy} that the players only need to know $x_0$ since everything else can be computed offline. 
If $(\bar\up, \bar\ue)$ is an equilibrium pair (producing the same payoff as the pair ($\up, \ue$)), then in our problem, it would imply that there is no need for communication between the remote sensor and the pursuer.  
However, $(\bar\up, \bar\ue)$ is not necessarily an equilibrium pair, even though it is derived from the equilibrium one $(\up, \ue)$. 
To best illustrate this, we consider the following example.
\begin{example} \label{example:preliminary}
Consider the pursuit-evasion game with $x_p \in \R^2$ and $x_e \in \R^2$ denoting the states of pursuer and the evader, respectively. 
They have the dynamics $\dot x_p =\up$ and $\dot x_e = u_e$. 
The state of the game is $x = [x_p\T, x_e\T]\T$ and the objective function is 
\[
\bar J = \int_{0}^1 \big(\frac{1}{4} \|\up\|^2 - \frac{1}{2}\|u_e\|^2\big) \d t + \|x_p(1) - x_e(1)\|^2.
\]
The initial states for the players are $x_p(0) = [0, 0]\T$ and $x_e(0) = [1, 0]\T$. 

For this example, we will explicitly write $I_n$ and $0_n$ to denote, respectively, the identity and the zero matrices of dimension $n$.

In this example, $A = 0_4$, $B = \begin{bmatrix} I_2 \\ 0_2\end{bmatrix}$ and $C = -\begin{bmatrix} 0 \\ I_2\end{bmatrix}$,  and $Q_f = \begin{bmatrix}
    ~~I_2 & - I_2 \\
    -I_2 & ~I_2
\end{bmatrix} $, $R_e = \frac{1}{2}$, $R_p = \frac{1}{4}$, $t_0=0$, and $t_f=1$.  

By solving \eqref{eq:riccati}, we obtain
\[
P(t) = \frac{1}{3-2t}\begin{bmatrix}
    ~~I_2 & -I_2\\
    -I_2 & ~I_2
\end{bmatrix}.
\]
The open-loop inputs, defined in \eqref{eq:barStrategy}, are found to be
\begin{align*}
    \bar\up = \Big[
    \frac{4}{3},~~ 0\Big]\T, \qquad  {\bar\ue = \Big[
    \frac{2}{3}, ~~0\Big]\T}.
\end{align*}
The resulting payoff from the strategy pair $(\bar\up, \bar\ue)$ is $\bar J = \frac{1}{3}$.

When the pursuer's strategy is fixed to $\bar\up$, the evader has an incentive to choose a strategy different from $\bar \ue$. 
To see this, let us pick an arbitrary strategy $\ue = [-c, ~~0]\T$ where $c>0$ is some constant.
Therefore, $x_p(1) = [\frac{4}{3},~~0]\T$ and $x_e(1) = [1-c, ~~0]\T$, and the payoff from this strategy is 
\begin{align*}
    \bar J & = \int_0^1 \bigg(\frac{1}{4}\Big(\frac{4}{3}\Big)^2 - \frac{1}{2}c^2\bigg)\d t  + \bigg(c + \frac{1}{3}\bigg)^2 \\
    & = \frac{1}{2}c^2 + \frac{2}{3}c + \frac{5}{9}. 
\end{align*}
The resulting payoff is not bounded from above as a function of $c$. 
Since the evader is aiming to maximize $\bar J$, it can drive the payoff to infinity.  \hfill $\triangle$
\end{example}

Example~\ref{example:preliminary} demonstrates that the pursuer may not commit to the open-loop equivalent of the equilibrium strategy without sacrificing performance. 
Moreover, as shown in this example, the payoff resulting from such an open-loop execution can be arbitrarily high (i.e., payoff diverging to infinity).  
While the difference in the payoffs resulting from $\up$ and $\bar\up$ is infinite in the chosen example, one might ask whether it is always the case for any such pursuit-evasion games represented by \eqref{eq:dyn}-\eqref{eq:Jbar}.  
Furthermore, one might be interested in the payoff when the pursuer has intermittent access to the state $x(t)$. 
These are some of the questions that we investigate in the rest of the paper. 
More precisely, we answer the following fundamental questions: (i) Does there exist an intermittent communication strategy between the pursuer and the remote sensor that can ensure the same payoff as the one obtained from continuous communication? 
(ii) If such an intermittent communication strategy exists, what is the minimum number ($N_{\min}$) of required communications, and what are the optimal time instances ($t_i$'s) for the communications? 
and finally, (iii) If the available number of communications is less than $N_{\min}$, how much will the payoff degrade?

\section{Game under Intermittent Communication} \label{sec:gameAnalysis}

Since the communication between the remote sensor and the pursuer is intermittent instead of continuous, and the total number of communications is bounded, the evader may have an incentive to deviate from its earlier equilibrium strategy of $\ue = R_e^{-1}C\T Px$. 
Without loss of generality, let us assume that the evader follows
\begin{align} \label{eq:evaderStrategy}
    \ue = R_e^{-1}C\T Px + w,
\end{align}
where $w(t) \in \R^{n_e}$ is an input vector that is chosen by the evader strategy $\mue$. 
The input $w(t)$ can depend on the state and it is also allowed to be a random variable. 
Allowing $w(t)$ to be a random variable implies that the evader can employ randomized/mixed strategies. 

Our first objective in this section is to investigate the existence of an intermittent communication strategy with the same payoff as the continuous communication strategy. 
To that end, let $\{t_i\}_{i=1}^{N}$ denote the $N$ communication instances where $t_i < t_{i+1} < t_f$ for all $i=1,\ldots, (N\!-\!1)$. 
For national brevity, we further introduce $t_{N+1} = t_f$ and the $0$-th communication instance is defined to be the initial time of the game $t_0$.

\subsection{Existence of an Intermittent Communication Strategy} 
To show the existence of an intermittent strategy that performs equally good as the continuous one, we first compute the payoff for a given set of communication instances $\Tp = \{t_i\}_{i=0}^{N+1}$.
For a given $\Tp$, let the pursuer follow the strategy 
\begin{align}
    &\dot{\hat{x}} = (A\!-BR_p^{-1}B\T P \!+ CR_e^{-1}C\T P)\hat{x}, \quad \hat{x}(t_i) = x(t_i), \label{eq:hatdynamics1} \\
    &\up = - R_p^{-1}B\T P \hat{x}.  \label{eq:upOptimal}
\end{align}
Later we will prove that \eqref{eq:upOptimal} is indeed the optimal strategy for the pursuer. 

Given the evader strategy to be \eqref{eq:evaderStrategy} and the pursuer strategy to be \eqref{eq:upOptimal}, we have
\begin{align} \label{eq:newDynamics}
    \dot{x} = A_1x - BR_p^{-1}B\T P \hat{x} + C w,
\end{align}
where $A_1 \triangleq A+C R_e^{-1}C\T P$. 
Let $e(t) \triangleq x(t) - \hat{x}(t)$ denote the difference between the true state $x$ and the pursuer's estimate $\hat{x}$ at time $t$. 
Using \eqref{eq:newDynamics} and \eqref{eq:hatdynamics1}, one may verify that, for all $t\in [t_{i}, t_{i+1} )$ and for all $i = 0, \dots, N$, 
\begin{align} \label{eq:e}
\begin{split}
&\dot{e} = A_1 e + C {w} \qquad t\in [t_{i}, t_{i+1} ) \\
&e(t_i) = 0.
\end{split}
\end{align}

Substituting \eqref{eq:upOptimal} and \eqref{eq:evaderStrategy} in \eqref{eq:Jbar_squared} yields
\begin{align} \label{eq:J_simplified}
    \bar J = &\|x_0\|^2_{P_0} + \int_{t_0}^{t_f} \big(  \| R_p^{-1}B\T P e\|^2_{R_p} - \|w\|^2_{R_e}\big) \d t.
\end{align}
Since $w$ is potentially a random process, we  consider the expected payoff
\begin{align} \label{eq:J_expected}
    J  = &\|x_0\|^2_{P_0} + \int_{t_0}^{t_f}\!\!  \E_{\sim \mu_e} \big(\|e\|^2_{P B R_p^{-1} B\T P} - \|w\|^2_{R_e}\big) \d t,
\end{align}
where the expectation is with respect to the randomness introduced by the evader strategy $\mue$. 
The expected cost \eqref{eq:J_expected} can be rewritten as 
\begin{align} \label{eq:J_expectedDecomposed}
    J  = &\|x_0\|^2_{P_0} + \sum_{i=0}^N \int_{t_i}^{t_{i+1}}\!\!  \E_{\sim \mu_e} \big(\|e\|^2_{P B R_p^{-1} B\T P} - \|w\|^2_{R_e}\big) \d t.
\end{align}
Since,  $e(t_i)$ resets to zero regardless of the evader's strategy, the choice of $w$ in the interval $[t_i, t_{i+1})$ does not affect the costs $\int_{t_j}^{t_{j+1}}\!\!  \E_{\sim \mu_e} \big(\|e\|^2_{P B R_p^{-1} B\T P} - \|w\|^2_{R_e}\big) \d t$ for any $j\ne i$. 
Therefore, \eqref{eq:J_expectedDecomposed} is more amenable than \eqref{eq:J_expected} for computing the optimal $w$ that maximizes $J$.

Let us now consider the evader's optimal control problem for the interval $[t_i, t_{i+1})$, which we denote by $\textbf{OP}_i$:
\begin{align*}
    \textbf{OP}_i \quad &\max \int_{t_i}^{t_{i+1}}\!\!  \E_{\sim \mu_e} \big(\|e\|^2_{P B R_p^{-1} B\T P} - \|w\|^2_{R_e}\big) \d t, \\
    &\text{subject to  }   \eqref{eq:e}.
\end{align*}
In the following theorem we characterize the optimal solution and the optimal value for $\textbf{OP}_i$. 

\begin{theorem} \label{thm:main1}
Let $M$ satisfy the following Riccati equation
\begin{align} \label{eq:M_Riccati}
    &\dot{M} + A_1\T M + MA_1 - PBR_p^{-1} B\T P - MCR_e^{-1} C\T M = 0, \nonumber \\
    & M(t_{i+1}) = 0,
\end{align}
where $A_1 = A + CR_e^{-1} C\T P$.
The optimal value of $\textbf{OP}_i$ is zero if and only if $M(t)$ does not have an escape time in the interval $[t_i, t_{i+1})$, and the corresponding optimal input is $w(t) = 0$. 
On the other hand, if $M(t)$ has an escape time in the interval $[t_i, t_{i+1})$, then the optimal value of $\textbf{OP}_i$ is unbounded from above. \hfill $\triangle$
\end{theorem}

\begin{proof}
    Let us first consider the case that $M(t)$ does not have an escape time in the interval $[t_i, t_{i+1})$. 
    In this case, one may use 
    \begin{align*}
        0 & = \int_{t_i}^{t_{i+1}}\frac{\d}{\d t}(\|e\|^2_M)\d t\\
          & =  \int_{t_i}^{t_{i+1}} \big( e\T (\dot M + A_1\T M + M A_1) e + 2 e\T M C w \big) \d t
    \end{align*}
     to rewrite 
    \begin{align*}
        \int_{t_i}^{t_{i+1}}\!\!  &\big(\|e\|^2_{P B R_p^{-1} B\T P} - \|w\|^2_{R_e}\big) \d t \\
        &= - \int_{t_i}^{t_{i+1}}\!\!  \|w - R_e^{-1}C\T M e\|^2_{R_e} \d t ~~\le 0.
    \end{align*}
    Since $w$ must be chosen to maximize the optimal value of $\textbf{OP}_i$, we conclude that $w = R_e^{-1}C\T M e$ is the optimal choice according to the last equation. 
    Substituting $w = R_e^{-1}C\T M e$ in \eqref{eq:e}, we obtain $e(t) = 0$ and consequently $w(t) = 0$. 

   On the other hand, when $M(t)$ has a finite escape time during the interval $[t_i, t_{i+1})$, one may construct an input $w$ to show that the optimal value is unbounded. 
    We omit this construction of $w$ due to page limitations.
\end{proof}

Theorem~\ref{thm:main1} implies that, as long as the communication instances $t_i$'s are chosen such that $M$ is well defined in each interval, the  optimal choice for the evader is $w = 0$, this is in line with what was observed for the sampled-data $H_\infty$ optimal control problem in \cite{bacsar1991game}.
Theorem~\ref{thm:main1} essentially shows that $J = \|x_0\|^2_{P_0}$, i.e., the payoff is exactly the same as the one from the continuous communication case. 
Theorem~\ref{thm:main1} also states that, if any of the communication instance is chosen such that $M$ has a finite escape time, then the payoff becomes infinity. 
Therefore, the intermittent communication can either perform as good as the continuous communication or it will perform arbitrarily bad, but nothing in between. 

\begin{corollary} \label{corr:necessary}
A necessary condition for an intermittent communication strategy to produce a finite payoff is to ensure $M(t)$ has a well defined solution for each interval $[t_i, t_{i+1})$.
\end{corollary}
\begin{proof}
    The proof follows directly from Theorem~\ref{thm:main1} and the preceeding discussion. 
\end{proof}

\begin{remark}
Given the $(i+1)$-th communication time $t_{i+1}$ and the terminal condition $M(t_{i+1}) = 0$, there exists a finite duration $\delta$ such that $M$ does not have a finite escape time in the interval $[t_{i+1}-\delta,~ t_{i+1})$.\footnote{This is due to the local Lipschitz property of the function $f(M)=-A_1\T M - MA_1 + PBR_p^{-1} B\T P + MCR_e^{-1} C\T M$. } 
Therefore, the inter-communication duration $(t_{i+1} - t_i)$ is lower bounded by $\delta$. 
This implies that the total number of communications is finite for a finite interval $[t_0, t_f)$.
\end{remark}

In summary, there always exists an intermittent communication strategy that produces the same payoff as the continuous communication case. 
Furthermore, the sensing instances must satisfy the necessary condition given in Corollary~\ref{corr:necessary}.
This necessary condition is related to the finite escape times of a certain Riccati equation. 
While there exists several intermittent strategies satisfying this necessary condition, next we investigate the optimal intermittent communication strategy that incurs the least number of communications. 
Before proceeding to the next section, we conclude this secstion by formally showing that the assumed strategy of the pursuer in \eqref{eq:upOptimal} is an equilibrium strategy. 

\begin{lemma}
The equilibrium strategy for the pursuer is:
\begin{align*}
    &\dot{\hat{x}} = (A\!-BR_p^{-1}B\T P \!+ CR_e^{-1}C\T P)\hat{x}, \quad \hat{x}(t_i) = x(t_i),  \\
    &\up = - R_p^{-1}B\T P \hat{x}(t).  
\end{align*}
\end{lemma}

\begin{proof}
    Without loss of any generality we assume that the communication instances are chosen such that they satisfy the necessary condition in Corollary~\ref{corr:necessary}; otherwise, the payoff is unbounded from above regardless of the pursuer's strategy.

    Since the evader's policy is to pick $w=0$ when the intermittent communication instances follow Corollary~\ref{corr:necessary}, we have 
    $\ue = - R_e^{-1}C\T P x$ and the payoff \eqref{eq:Jbar_squared} becomes
    \begin{align*}
        \bar J = \|x_0\|^2_{P_0} + \int_{t_0}^{t_f} \|\up + R_p^{-1}B\T P x\|^2_{R_p}\d t.
    \end{align*}
    At this point, one may verify that $\hat{x}(t) = x(t)$ for all $t$, and hence, $\up = -R_p^{-1}B\T P \hat x$ is the equillibrium strategy.
\end{proof}

\subsection{Optimal Communication Instances}

The discussion in this section will  focus on computing the optimal communication instances.
To that end, we shall discuss the finite escape times of $M$ defined in \eqref{eq:M_Riccati}.
Two major challenges in studying the escape times of \eqref{eq:M_Riccati} are: (i) $A_1$ is time varying since it depends on $P$, and (ii) $M$ depends on $P$, which follows another Riccati equation. 

To overcome these challenges we consider the matrix $\Pi = M - P$. 
Since $P$ does not have a finite escape time (due to Assumption~\ref{asm:controllability}), $M$ will have a finite escape time if and only if $\Pi$ does.
Next we verify that $\Pi$ also follows a Riccati equation: 
\begin{align}
    \dot{\Pi} & = \dot M - \dot P\nonumber \\ 
      &= - (A_1\T M + MA_1 - PBR_p^{-1} B\T P - MCR_e^{-1} C\T M) \nonumber \\
      &~ + A\T P + PA + Q + P(CR_e^{-1}C\T - BR_p^{-1} B\T)P \nonumber \\
      &= -A\T \Pi - \Pi A + Q + \Pi CR_e^{-1} C\T \Pi. \nonumber 
\end{align}

Notice that the dynamics equation of $\Pi$ does not depend on $P$ anymore, and all the matrices involved the equation of $\dot \Pi$ are time invariant. 
This form is important since we may readily use the results from \cite{sasagawa1982finite} that studies the escape time of Riccati equations of this form.
Next, we will study the finite escape time of $\Pi$ in the interval $[t_i, t_{i+1})$ with the boundary condition $\Pi(t_{i+1}) = - P(t_{i+1})$. 
The following theorem provides the optimal communication instances. 

\begin{theorem} \label{thm:main2}
    The $i$-th triggering instance is the escape time for $\Pi$ where 
    \begin{align*}
    &\dot{\Pi} + A\T \Pi + \Pi A - Q - \Pi CR_e^{-1} C\T \Pi = 0, \qquad t< t_{i+1},\\
    &\Pi (t_{i+1}) = - P(t_{i+1}), 
\end{align*}
   where $t_{N+1} = t_f$. 
\end{theorem}

\begin{proof}
   Let $\{t_i\}_{i=1}^{N}$ be the set of communication times obtained from Theorem~\ref{thm:main2} and let $\{t_i'\}_{i=1}^{N'}$ be an arbitrary sequence of communication instances that satisfy the necessary condition in Corollary~\ref{corr:necessary}.
   We prove this theorem by contradiction and, to that end, we assume that $N'< N$.
   Notice that $t_{N+1} = t_{N' + 1}' = t_f$. 
   Since $t_N$ is the escape time for the interval $[t_N, t_f)$, then we must have $t_N\le t_{N'}'$. 
   Now, starting from the terminal condition $\Pi(t_{N'}') = - P(t_{N'}')$, one may verify that $\Pi(t_N) \preceq P(t_N)$ and therefore, $\Pi$ will have an escape time before $t_{N-1}$ since $t_{N-1}$ was the escape time with the boundary condition $\Pi(t_N) = P(t_N)$. 
   Therefore, we conclude $t_{N-1} \le t_{N'-1}'$.
   Following an inductive argument, one obtain $t_1 \le t_1'$ and consequently, $N' \ge N$. Thus, the proof is complete.
\end{proof}

Notice that the triggering instances are computed backwards and the duration $(t_{i+1} - t_{i})$ is maximized while finding $t_i$ for a given $t_{i+1}$. 
Since $t_0$ is fixed and $t_1$ is the first triggering instance, we are guaranteed that $\Pi$ does not have a finite escape time in the interval $(t_0, t_1]$.  
Therefore, unless $t_0$ itself is an escape time, there is some slack in the choice of $t_1$, and $t_1$ can be increased to $t_1'$ without introducing an escape time within the new interval $(t_0, t_1']$. 
Given this new $t_1'$, one may increase $t_2$ to $t_2'$ while ensuring  $(t_1', t_2']$ does not contain an escape time.
Therefore, the optimal communication instances are non-unique unless $t_0$ is the escape time for $\Pi$ with the boundary condition $\Pi(t_1) = - P(t_1)$. 
Although there is non-uniqueness in the actual communication instances, the total number of required communication is unique.

\subsection{Evader's Dilemma}

At time $t_i$ (or, at an arbitrary time $t$ in general), the evader does not know when/whether the pursuer is going to request for communications. 
Therefore, if the evader picks $w=0$ for the interval $[t_i, t_{i+1})$ and the pursuer does not request for a communication, then the evader has lost the opportunity that would have given the evader a much higher (theoretically infinite) payoff if it did not select $w=0$. 
On the other hand, if the evader picked a non-zero $w$ and the pursuer  communicated with the sensor, then the evader will have incurred a loss in its payoff. 
Therefore, the evader has to make a decision first on whether it should use a non-zero $w$ without the pursuer having to commit to the next communication instance. 
This is an interesting dilemma from the evader's side where it has to pick between a high-risk-high-reward (i.e., the evader picks a nonzero arbitrarily high $w$) and a no-risk-no-reward (i.e., evader picks $w=0$) strategy. 
This dilemma will not occur in a slightly different scenario where the remote sensor can continuously sense the evader and the sensor imitates the communication instead of the pursuer requesting for it. 
In this case, the evader knows that if it picks a nonzero $w$, the sensor will communicate state at appropriate $t_i$'s and this will result in a worse payoff for the evader. 
Therefore, even in absence of a communication from the sensor, the pursuer is ensured that the evader is not able to get a payoff higher than $\|x_0\|^2_{P_0}$ and the pursuer is safe to continue with its strategy $\up = -R_p^{-1}B\T P \hat x$ without having to reset the value of $\hat x$. 
This implies an implicit communication in absence of a physical communication, which is in line with what was found for a linear-quadratic optimal control problem in \cite{maity2020minimal}.

    \begin{figure}
      \begin{subfigure}[c]{0.45 \linewidth}
      \includegraphics[trim = 10 15 20 30, clip,  scale = 0.36]{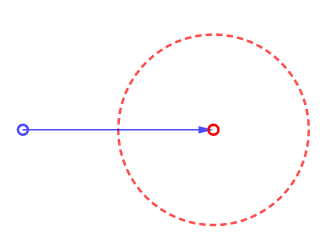}
      \end{subfigure}
      \begin{subfigure}[c]{0.35 \linewidth}
      \includegraphics[trim = 12 10 10 20, clip, scale = 0.4]{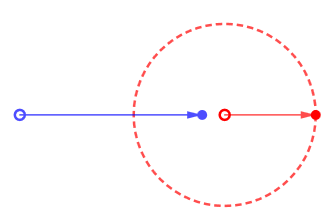}
      \put(-3,30) {$x_{\rm opt}$}
      \end{subfigure}

      \begin{subfigure}[c]{0.58 \linewidth}
      \includegraphics[trim = 12 15 25 20, clip, scale = 0.4]{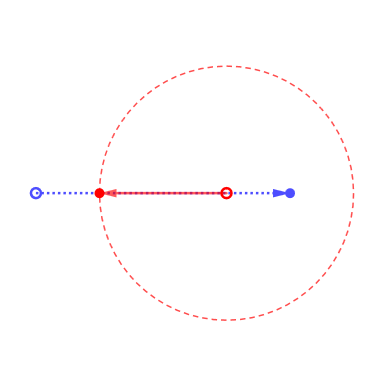}
      \put(-105,55) {$x_{\rm risky}$}
      \end{subfigure}
      \begin{subfigure}[c]{0.35 \linewidth}
      \includegraphics[trim = 12 15 10 20, clip, scale = 0.4]{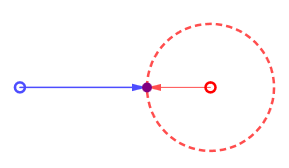}
      \put(-58,18) {$x_{\rm risky}$}
      \end{subfigure}
      \caption{\small In all the subfigures, hollow blue (red) points represent the initial location of the pursuer (evader). The solid blue (red) points represent the pursuer's (evader's) location at some time $t_1$.
      The dotted red circles denote the boundary of the evader's reachability region under the constraint that the control effort (measured by $\int_{t_0}^{t_1}\|\ue\|^2_{R_e}$) is less than or equal to the effort from the optimal controller.
      (Top left): $t_1 = \frac{3}{4}$ and the pursuer is at the evader's initial location. 
      (Top-right): $t_1 < \frac{3}{4}$ and the optimal location for the evader to be is at $x_{\rm opt}$.
      (Bottom left): $t_1 > \frac{3}{4}$ and the optimal location for the evader to be at $x_{\rm risky}$. 
      (Top-right): $t_1 = \frac{1}{2}$ the location $x_{\rm risky}$ coincides with the pursuer's location.
      }
      \label{fig:stateSpace}
  \end{figure}

  \subsection{Example~\ref{example:preliminary} Revisited}

In this section, we revisit Example~\ref{example:preliminary} and compute the communication instances according to Theorem~\ref{thm:main2}. 
By solving the Riccati equation for $\Pi$, we obtain 
\begin{align*}
    \Pi(t) = \frac{1}{3-4t_{i+1}+2t}\begin{bmatrix}
    -I_2 & ~~I_2\\
    ~I_2 & -I_2
\end{bmatrix}. \quad \forall t\in (t_i, t_{i+1}].
\end{align*}
Therefore, the $i$-th communication time is found by solving $3-4t_{i+1}+2t_i = 0$. 
Given $t_{N+1}=t_f = 1$, we obtain, $t_N = \frac{4t_{N+1}-3}{2} = \frac{1}{2}$. 
Similarly, given $t_N = \frac{1}{2}$, we obtain $t_{N-1} = \frac{4t_N - 3}{2} = - \frac{1}{2} < t_0$. 
Therefore, only one communication is needed and the communication occurs at time $t_1 = \frac{1}{2}$. 

As discussed after Theorem~\ref{thm:main2}, the communication instances are non-unique. In our example, one may choose any $t_1$ such that $t_1< \frac{3}{4}$. 
This is obtained by plugging in $t_0 = 0$ in the equation $3-4t_1 + 2t_0 > 0$. 

The conditions $t_1< \frac{3}{4}$ or $t_1 = \frac{1}{2}$ are obtained by solving an escape time problem and they do not provide much physical insights. 
To understand what these conditions physically imply in the context of this example, we use a graphical representation of the system in Fig.~\ref{fig:stateSpace}, where the blue hollow/filled dots represent the pursuer and the red ones representing the evader. 
Recall that the optimal (open-loop) input for the evader was $\ue = [\frac{2}{3},~0]\T$ and therefore, the control effort used in the interval $[0, t_1]$ is $\int_{0}^{t_1}\frac{1}{2}\|\ue\|^2 = \frac{2t_1}{9}$. 
    Let us consider the set of all control functions with a maximum control effort of $\frac{2t_1}{9}$ in the interval $[0,t_1)$, i.e., consider $\mathcal{U} =\{\ue~|~ \int_{0}^{t_1}\frac{1}{2}\|\ue\|^2 \le \frac{2t_1}{9}\}$.
    Given a $t_1$, the reachability set of the evader becomes a circle with radius $\frac{2t_1}{3}$ when the evader's input is restricted to the set $\mathcal{U}$. 
    These reachability sets (boundaries) are shown in the subfigures with dashed red circles.

    Notice that, for $\up = \big[
    \frac{4}{3},~~ 0\big]\T$, $t_1 = \frac{3}{4}$ is the time when the defender reaches the initial location of the evader, which is at $[1, ~0]\T$. 
    If the pursuer communicates with the sensor at this very moment (i.e., at $t=\frac{3}{4}$), then the payoff starting from this time is going to be the same regardless of where the evader is on the boundary of the dotted red circle. 
    On the other hand, if the pursuer communicates before time $\frac{3}{4}$, as illustrated in the top-right subfigure of Fig.~\ref{fig:stateSpace}, then the the evader must be at $x_{\rm opt}$ to start the next segment of the game, as also shown in that  subfigure. 

    If the pursuer communicates after $t=\frac{3}{4}$, as shown in the bottom-left subfigure, the optimal strategy for the evader is to go behind the pursuer and be at $x_{\rm risky}$, since any other point on the reachability circle will have a smaller distance from the pursuer. 
    The evader's strategy in Example~\ref{example:preliminary} was constructed based on this observation.

    On the other hand, at time $t=\frac{1}{2}$, the pursuer's location intersects with the evader's reachability circle at the point $x_{\rm risky}$. 
    Therefore, if the evader decided to go behind the defender, this would be the most vulnerable point in time. 

    Although the evader's input was restricted to the set $\mathcal{U}$ for the above discussion, a similar observation is noted also when this restriction is omitted.

\section{Conclusion and Discussions} \label{sec:conclusion}

In this work, we considered a class of linear-quadratic pursuit-evasion games where the pursuer relies on a remote sensor to measure the current state of game. 
The pursuer intermittently communicates with the sensor and the total number of communication is minimized. 
The optimal communication instances are found by solving for the finite escape times of a certain Riccati equation.
The optimal communication instances are, in general, non-unique. 
The evader is faced with a dilemma between taking a high-risk-high-reward and a no-risk-no-reward strategy. 

In this work, we assumed that the remote sensor can perfectly sense the game state and the communication between the sensor and the pursuer is ideal (i.e., no delay, no packet loss, and no transmission noise).
It would be interesting to study the problem with the communication being suffered from random delays. 
In that case, the pursuer may not receive the measurement in time to satisfy the escape time condition. 
Since there is generally some slack in the choice of the sensing instances, the pursuer may utilize this slack to ensure that the measurement is received in time, although such an approach will only work with certainty if the communication delay is bounded by the available slack. 
Alternatively, the pursuer needs to increase the number of communications. 
It is not obvious how to choose the communication instances in presence of stochastic delays. 
Likewise, analyzing the effects of packet dropoutson the payoff is also an interesting research direction.

We notice that the evader is able to achieve an arbitrarily high payoff when the pursuer lets $\Pi$ to have a finite escape time. 
However, in order to achieve this arbitrarily high payoff, the evader needs apply a control input with arbitrary high magnitude, which is unrealistic from a practical point of view. 
Therefore, perhaps adding an upper bound on the magnitude of $\ue(t)$ (and $\up(t)$) is a more realistic scenario for this problem. 
In this case, the pursuer can ensure a finite payoff even when $\Pi$ has finite escape times. 
The received payoff will degrade if $\Pi$ is allowed to have finite escape times and this degradation in the payoff will be related to the inter-communication durations and the upper bound on the magnitude of $\ue$.


\bibliographystyle{IEEEtran}
\bibliography{reference}

\begin{thebibliography}{10}
\providecommand{\url}[1]{#1}
\csname url@samestyle\endcsname
\providecommand{\newblock}{\relax}
\providecommand{\bibinfo}[2]{#2}
\providecommand{\BIBentrySTDinterwordspacing}{\spaceskip=0pt\relax}
\providecommand{\BIBentryALTinterwordstretchfactor}{4}
\providecommand{\BIBentryALTinterwordspacing}{\spaceskip=\fontdimen2\font plus
\BIBentryALTinterwordstretchfactor\fontdimen3\font minus
  \fontdimen4\font\relax}
\providecommand{\BIBforeignlanguage}[2]{{%
\expandafter\ifx\csname l@#1\endcsname\relax
\typeout{** WARNING: IEEEtran.bst: No hyphenation pattern has been}%
\typeout{** loaded for the language `#1'. Using the pattern for}%
\typeout{** the default language instead.}%
\else
\language=\csname l@#1\endcsname
\fi
#2}}
\providecommand{\BIBdecl}{\relax}
\BIBdecl

\bibitem{isaacs1999differential}
R.~Isaacs, \emph{Differential games: a mathematical theory with applications to
  warfare and pursuit, control and optimization}.\hskip 1em plus 0.5em minus
  0.4em\relax Courier Corporation, 1999.

\bibitem{yan2016multi}
C.~Yan and T.~Zhang, ``Multi-robot patrol: A distributed algorithm based on
  expected idleness,'' \emph{International Journal of Advanced Robotic
  Systems}, vol.~13, no.~6, 2016.

\bibitem{robotics9020047}
T.~Alam and L.~Bobadilla, ``Multi-robot coverage and persistent monitoring in
  sensing-constrained environments,'' \emph{Robotics}, vol.~9, no.~2, 2020.

\bibitem{inproceedings}
E.~García, A.~Von~Moll, D.~Casbeer, and M.~Pachter, ``Strategies for defending
  a coastline against multiple attackers,'' in \emph{58th Conference on
  Decision and Control}.\hskip 1em plus 0.5em minus 0.4em\relax IEEE, 2019, pp.
  7319--7324.

\bibitem{sun2017multiple}
W.~Sun, P.~Tsiotras, T.~Lolla, D.~N. Subramani, and P.~F. Lermusiaux,
  ``Multiple-pursuer/one-evader pursuit--evasion game in dynamic flowfields,''
  \emph{Journal of guidance, control, and dynamics}, vol.~40, no.~7, pp.
  1627--1637, 2017.

\bibitem{oyler2016pursuit}
D.~W. Oyler, P.~T. Kabamba, and A.~R. Girard, ``Pursuit--evasion games in the
  presence of obstacles,'' \emph{Automatica}, vol.~65, pp. 1--11, 2016.

\bibitem{bhattacharya2010existence}
S.~Bhattacharya and S.~Hutchinson, ``On the existence of {N}ash equilibrium for
  a two-player pursuit—evasion game with visibility constraints,'' \emph{The
  International Journal of Robotics Research}, vol.~29, no.~7, pp. 831--839,
  2010.

\bibitem{aleem2015self}
S.~A. Aleem, C.~Nowzari, and G.~J. Pappas, ``Self-triggered pursuit of a single
  evader,'' in \emph{54th Conference on Decision and Control}.\hskip 1em plus
  0.5em minus 0.4em\relax IEEE, 2015, pp. 1433--1440.

\bibitem{maity2016strategies}
D.~Maity and J.~S. Baras, ``Strategies for two-player differential games with
  costly information,'' in \emph{13th International Workshop on Discrete Event
  Systems}.\hskip 1em plus 0.5em minus 0.4em\relax IEEE, 2016, pp. 211--216.

\bibitem{huang2021defending}
Y.~Huang, J.~Chen, and Q.~Zhu, ``Defending an asset with partial information
  and selected observations: A differential game framework,'' in \emph{60th
  Conference on Decision and Control}.\hskip 1em plus 0.5em minus 0.4em\relax
  IEEE, 2021, pp. 2366--2373.

\bibitem{durham2010distributed}
J.~W. Durham, A.~Franchi, and F.~Bullo, ``Distributed pursuit-evasion with
  limited-visibility sensors via frontier-based exploration,'' in
  \emph{International Conference on Robotics and Automation}.\hskip 1em plus
  0.5em minus 0.4em\relax IEEE, 2010, pp. 3562--3568.

\bibitem{gerkey2006visibility}
B.~P. Gerkey, S.~Thrun, and G.~Gordon, ``Visibility-based pursuit-evasion with
  limited field of view,'' \emph{The International Journal of Robotics
  Research}, vol.~25, no.~4, pp. 299--315, 2006.

\bibitem{bagchi1981linear}
A.~Bagchi and G.~J. Olsder, ``Linear-quadratic stochastic pursuit-evasion
  games,'' \emph{Applied mathematics and optimization}, vol.~7, no.~1, pp.
  95--123, 1981.

\bibitem{behn1968class}
R.~Behn and Y.-C. Ho, ``On a class of linear stochastic differential games,''
  \emph{IEEE Transactions on Automatic Control}, vol.~13, no.~3, pp. 227--240,
  1968.

\bibitem{rhodes1969differential}
I.~Rhodes and D.~Luenberger, ``Differential games with imperfect state
  information,'' \emph{IEEE Transactions on Automatic Control}, vol.~14, no.~1,
  pp. 29--38, 1969.

\bibitem{maity2016optimal}
D.~Maity and J.~S. Baras, ``Optimal strategies for stochastic linear quadratic
  differential games with costly information,'' in \emph{55th Conference on
  Decision and Control}.\hskip 1em plus 0.5em minus 0.4em\relax IEEE, 2016, pp.
  276--282.

\bibitem{maity2017linear}
D.~Maity, A.~Anastasopoulos, and J.~S. Baras, ``Linear quadratic games with
  costly measurements,'' in \emph{56th Conference on Decision and
  Control}.\hskip 1em plus 0.5em minus 0.4em\relax IEEE, 2017, pp. 6223--6228.

\bibitem{maity2017asymptotic}
D.~Maity and J.~S. Baras, ``Asymptotic policies for stochastic differential
  linear quadratic games with intermittent state feedback,'' in \emph{25th
  Mediterranean Conference on Control and Automation}.\hskip 1em plus 0.5em
  minus 0.4em\relax IEEE, 2017, pp. 117--122.

\bibitem{bacsar1991game}
T.~Ba{\c{s}}ar, ``Game theory and {$H_{\infty}$}-optimal control: The
  continuous-time case,'' in \emph{Differential Games—Developments in
  Modelling and Computation: Proceedings of the Fourth International Symposium
  on Differential Games and Applications, Helsinki University of Technology,
  Finland}.\hskip 1em plus 0.5em minus 0.4em\relax Springer, 1991, pp.
  171--186.

\bibitem{cai2020distributed}
X.~Cai, F.~Xiao, and B.~Wei, ``A distributed strategy-updating rule with
  event-triggered communication for noncooperative games,'' in \emph{39th
  Chinese Control Conference}.\hskip 1em plus 0.5em minus 0.4em\relax IEEE,
  2020, pp. 4747--4752.

\bibitem{liu2023predefined}
J.~Liu and P.~Yi, ``Predefined-time distributed {Nash} equilibrium seeking for
  noncooperative games with event-triggered communication,'' \emph{IEEE
  Transactions on Circuits and Systems II: Express Briefs}, 2023.

\bibitem{molin2009lqg}
A.~Molin and S.~Hirche, ``On {LQG} joint optimal scheduling and control under
  communication constraints,'' in \emph{Proceedings of the 48h IEEE Conference
  on Decision and Control held jointly with 28th Chinese Control
  Conference}.\hskip 1em plus 0.5em minus 0.4em\relax IEEE, 2009, pp.
  5832--5838.

\bibitem{maity2019optimal}
D.~Maity and J.~S. Baras, ``Optimal event-triggered control of nondeterministic
  linear systems,'' \emph{IEEE Transactions on Automatic Control}, vol.~65,
  no.~2, pp. 604--619, 2019.

\bibitem{sasagawa1982finite}
T.~Sasagawa, ``On the finite escape phenomena for matrix riccati equations,''
  \emph{IEEE Transactions on Automatic Control}, vol.~27, no.~4, pp. 977--979,
  1982.

\bibitem{maity2020minimal}
D.~Maity and J.~S. Baras, ``Minimal feedback optimal control of
  linear-quadratic-{G}aussian systems: No communication is also a
  communication,'' \emph{IFAC-PapersOnLine}, vol.~53, no.~2, pp. 2201--2207,
  2020.

\end{thebibliography}

\end{document}